\documentclass[11pt, oneside,reqno]{amsart} 
 \usepackage{amsaddr}
\usepackage{graphicx} 
\usepackage{cite} 
\usepackage{algorithm} 
\usepackage[noend]{algpseudocode}
\usepackage{multirow}
\usepackage{xcolor} 
\newtheorem{corollary}{Corollary}
\newtheorem{theorem}{Theorem}
\newtheorem{lemma}{Lemma}
\theoremstyle{definition}  
   
\newcommand{\ii}{\mathrm{i}}  
\newcommand{\SL}{SL(3,\mathbb{Z})} 
  
\newcommand{\eqt}[1]{\begin{align} #1\end{align}} 
\newcommand{\eqtx}[1]{\begin{align*} #1\end{align*}} 
\newcommand{\mat}[1]{{\begin{bmatrix} #1 \end{bmatrix}}}  
\newcommand{\mround}[1]{\left \lfloor #1 \right \rceil } 
\newcommand{\bff}[1]{\textbf{#1}  } 
\newcommand{\diag}{\textrm{diag}}

\title[PBCs for Uniaxial Flow]{Simple Periodic Boundary Conditions for Molecular Simulation of Uniaxial Flow}
\author{Matthew Dobson and Abdel Kader Geraldo} 
\address{Department of Mathematics, University of Massachusetts Amherst, 
Amherst, MA 01003}
\email{dobson@math.umass.edu}
\date{\today}

\begin{document}
\begin{abstract} 
We present  rotating  periodic boundary conditions  (PBCs) for the simulation of nonequilibrium molecular dynamics (NEMD) under uniaxial stretching flow (USF) or biaxial stretching flow (BSF). Such nonequilibrium flows need specialized PBCs since the simulation box deforms with the background flow.   
The technique builds on previous models using one or  lattice remappings, and is simpler than the PBCs developed for the general three dimensional flow.
For general three  dimensional flows, Dobson \cite{Dobson} and Hunt \cite{Hunt} proposed schemes which are not time-periodic since they use more than one automorphism remapping.  This paper presents
a single automorphism remapping PBCs for USF and BSF which is time periodic up to a rotation matrix and has a better minimum lattice spacing properties.
\end{abstract}

\maketitle 
\section{Introduction} 
Nonequilibrium molecular dynamics (NEMD)~\cite{Evans, todd17} techniques are one tool used to study molecular fluids under steady flow, and for instance, some of recent applications can be found in \cite{Lang-phil,Oconnor,Oconnor2,Nicholson,Oliveira,Nishioka-Akihiro,Templeton,Menzel-Daivis-Todd,Ewen,Li-Zhen}. However, there are special challenges
in formulating the periodic boundary conditions  (PBCs) in
the nonequilibrium setting~\cite{DAIVIS,XuWen,Baranyai,todd1998,todd2000}.

We consider a molecular simulation where the particles have an average flow  consistent 
with a homogeneous background flow matrix 
$A=\nabla \textbf{u} \in \mathbb{R}^{3 \times 3}.$  This flow is used to simulate the
micro-scale motion of a fluid with local strain rate $\nabla \textbf{u}.$ 
We denote the coordinates of the simulation box via three linearly independent vectors
coming from the origin, and we write the vectors as the columns of the matrix 
\eqtx{ 
L_t=\mat{\mathbf{v}_t^1 &\mathbf{v}_t^2 &\mathbf{v}_t^3}\in \mathbb{R}^{3 \times 3}, \qquad t \in [0,\infty).
} 
To be consistent with the
background flow, a particle with a coordinates $(\textbf{q,v} )$, where $\textbf{q}$ is the 
position and $\textbf{v}$ is   the velocity, has images with coordinates 
$(\bff{q}+L_t \bff{n},\bff{ v}+A L_t  \bff{n})$,  where $\bff{n} \in \mathbb{Z}^{3  }$ 
are triples of integers.  Since the image velocity is the time derivative of its position we have
\eqtx{
\frac{d}{dt}(\bff{q}+L_t \bff{n})=\bff{ v}+A L_t  \bff{n},
}
which implies that the simulation box deforms with the flow
\eqtx{
\frac{d}{dt} L_t  = A L_t .
}
If the initial lattice $L_0$ is not chosen carefully,  
the resulting lattice deformation
\eqtx{ 
  L_t  = e^{A t} L_0
}
can become degenerate and lead to a particle and some of its images becoming arbitrarily close. We want to ensure that the  minimum distance between a particle and its images is nonzero
for all time, 
\eqt{\label{min-dist}
d=  \inf_{\begin{array}{c}\scriptstyle \mathbf{n}\in \mathbb{Z}^3 \backslash {0}\\
\scriptstyle  t\in \mathbb{R}\geq 0\end{array}}  ||  {L}_t \mathbf{n}||_2>0.
}  
This is necessary in order to have long-time stable periodic boundary conditions for  NEMD flows. 

We consider a class of PBCs based on remapping the simulation box at various times during the simulation
by choosing a new set of basis vectors for the   lattice $L_t$ that describes the simulation box.  This remapping is 
called a lattice automorphism
and can be represented as a $3 \times 3$ integer matrix with determinant one.  This was
first used for the case of shear flow by Lees-Edwards~\cite{LE}  and was then extended to the case of planar elongational flow
 by   Kraynik and Reinelt (KR)~\cite{KR}. Those  algorithms 
result in remappings which are periodic in time, though KR showed that a 
time-periodic remapping to the original simulation box using such matrices is impossible for general three dimensional flows.  
 Dobson~\cite{Dobson} and  Hunt~\cite{Hunt}   developed  PBCs for general three dimensional diagonalizable flow using similar remapping  technique to the KR scheme. Those schemes use more than one automorphism matrix and result in a remapping that is not time periodic.  In this paper we present a rotating box algorithm applicable to uniaxial stretching flow (USF) and biaxial stretching flow (BSF)  which features advantageous properties. Namely,  we will show that using the class of automorphism matrices that has a pair of complex conjugate eigenvalues, we can construct a single remapping matrix algorithm  which is   time periodic up to a rotation matrix and whose minimum distance~\eqref{min-dist} is larger than those of the GenKR and Hunt algorithms.
 
The outline of this paper is as follows. Section~\ref{sect2} gives the background for PBCs especially shear flow, planar elongational flow, and general three dimensional flows. Section~\ref{sect3} presents the rotating box algorithm, and Section~\ref{sect4} gives the prove that the deformed lattice obtained is not time periodic. Section~\ref{sect5} compares the rotating box algorithm with the existing three dimensional flow PBCs.

\section{Background} 
\label{sect2}
  
In this section, we   give a description of the existing remapping PBCs, starting with the two dimensional flows, especially,   shear  flow and   planar elongational flow. In the case of   three dimensional flows, the generalized KR (GenKR) algorithm developed by Dobson and   Hunt are presented.
All the algorithms follow the same procedure: given a background flow $A$, for each time $t$  we find the appropriate integer power   of the chosen automorphism matrix (or matrices) to remap the lattice $L_t$.
 
\subsection{Shear Flow}
We first consider the shear flow case where the background  matrix $A$ is given by
 \eqtx{
A=\mat{
0 & \epsilon & 0 \\
0 & 0 & 0 \\
0 & 0 & 0
} .
} 
At time $t$, the lattice
is given by
\eqtx{ 
L_t= \mat{1& t \epsilon& 0\\0&1&0\\0&0&1}  L_0 \textrm{ where } L_0= \mat{1&  0& 0\\0&1&0\\0&0&1}.
}
A highly sheared box makes the computation of interparticle interactions
more difficult, however this problem can be overcome by looking at the geometry
of shears that are integer multiples of the box length.
The Lees-Edwards (LE) boundary conditions~\cite{LE}  is used to
prevent  the simulation box from becoming too deformed.  
Whenever the simulation time is an integer multiple of the inverse shear rate, $t_n = n\epsilon^{-1},$
the simulation box is sheared by $n$ box lengths.  We remap the 
simulation box with the matrix  
\eqtx{
M^n=\mat{1& -1& 0\\0&1&0\\0&0&1}^n=\mat{1& -n& 0\\0&1&0\\0&0&1}, \quad n\in \mathbb{Z}
}
such that at a time $t$, the simulation box lattice is  
\eqtx{L_t M^n=\mat{1& t\epsilon-n& 0\\0&1&0\\0&0&1}, n \in \mathbb{Z}. } 
Since $M$ is an integer matrix with determinant equal to one, that is, 
$M \in \SL,$ the matrices $L_t$ and $L_t M^n$ generate the same lattice.
Throughout the simulation,  we choose  $n=-\mround{t \epsilon}$    so
that the stretch is at most half of the simulation box, 
and that this remapping process is time-periodic with period $t^{\ast}=\frac{1}{\epsilon}$, where $\mround{x}$ denote $x$  rounded to nearest integer.
 \subsection{Planar Elongational Flow}\label{PEF}
 Here, the background flow matrix is 
  \eqtx{
A=\mat{
\epsilon &0 & 0 \\
0 & -\epsilon & 0 \\
0 & 0 & 0
} ,
}  meaning that
the simulation box elongates in the $x$ direction and shrinks in the $y$ direction of the standard coordinate plane.   To treat this case, KR proposed the use of a diagonalizable automorphism  matrix  ${M}\in \SL$ that has the form 
\eqtx{
M V=V \Lambda, \quad \Lambda= \mat{
\lambda &0 & 0 \\
0 & \lambda^{-1}  & 0 \\
0 & 0 &1
}  
, \quad \lambda > 0,  \quad \lambda \neq 1,
} 
to remap the simulation box. 
We consider  the initial lattice   $L_0=V^{-1}$  so that  at time $t$ when we  
apply $M^n$ to the 
lattice basis vectors 
\eqtx{
L_{t } M^n= e^{ t A} L_0 M^{n} =  e^{t \epsilon D} \Lambda^n  V^{-1}= e^{ A_t }V^{-1} \textrm{, where }D= \mat{
1 &0 & 0 \\
0 &-1 & 0 \\
0 & 0 &0
},
} 
and $A_t=(t \epsilon+ n\log(\lambda) )D$.
 Letting $n=-\mround{t \frac{ \epsilon}{\log(\lambda)}}$, the  stretch of the flow  $A_t$
   remains bounded during the simulation, and in addition, it is time periodic with period $ t_*=\frac{\log(\lambda)}{\epsilon}.$  For instance, 
\eqtx{
M=\mat{
2 &-1 & 0 \\
-1 & 1 & 0 \\
0 & 0 & 1
} 
}
is an example of matrix which gives a good minimum distance between a particle and its images.
\subsection{General three-dimensional (3D) flow PBCs} 

For a general 3D flow
\eqtx{
A= \mat{
\epsilon_1 &0 & 0 \\
0 &\epsilon_2  & 0 \\
0 & 0 &-\epsilon_1-\epsilon_2
}  ,
}  
Dobson  and Hunt  proposed equivalent algorithms to control the deformation.
\subsubsection{Dobson's Approach}
In \cite{Dobson}, the author develops PBCs which use two commutative automorphism matrices $M_1, M_2 \in \SL$ which have positive eigenvalues for the remapping of the simulation box.  
 Since the matrices are commutative, they are simultaneously diagonalizable, $M_i V= V \Lambda_i.$ An example of the pair of the automorphism matrices are
 \eqtx{
 M_1=\mat{
1 & 1 & 1 \\
1  & 2 & 2 \\
1 & 2 & 3
} \quad \textrm{and }
M_2=\mat{
2 & -2 & 1 \\
-2 & 3 & -1\\
1 & -1 & 1
}.
 }
 The algorithm requires that 
the diagonal of the logarithm of the eigenvalue matrices $ 
 \Hat{\omega}_i = \log (\Lambda_i)$  
must be linearly independent, thus there exists  $\delta_i\in \mathbb{R}$ solving $  A  = \delta_1 \Hat{\omega}_1+\delta_2\Hat{\omega}_2$. 
Now by considering the initial lattice $L_0=V^{-1}$ and picking $n_i=-\mround{t \delta_i}$, we remark that the remapping of the simulation box with $M_1^{n_1} M_2^{n_2}$ results in   the remapped lattice
 \eqtx{ 
 \Tilde{L}_t={L}_t M_1^{n_1} M_2^{n_2} &= e^{A t} L_0 M_1^{n_1} M_2^{n_2} =e^{t A } \Lambda_1^{n_1} \Lambda_2^{n_2} V^{-1} = e^{ A_t  } V^{-1} 
   ,} 
where  the remaining stretch matrix
\eqtx{
A_t= t A + n_1 \Hat{\omega}_1+n_2\Hat{\omega}_2=(t\delta_1-\mround{t\delta_1}) \Hat{\omega}_1+(t\delta_2-\mround{t\delta_2})\Hat{\omega}_2,
 }
 is clearly bounded for every time $t$. Thus the minimum distance of the remapped lattice is bounded away from zero during the entire simulation.

 \subsubsection{Hunt's Approach}
 Hunt's approach is  similar to Dobson's, using the Lenstra-Lenstra-Lov\'{a}sz ( $LLL$)~\cite{LLL}
  in place of a second automorphism matrix. 
 As convention in this paper, we will describe   Hunt's algorithm using    column vectors instead of the row vectors used in the original paper. 
 In fact, Hunt's PBCs consists of remapping the simulation box with the automorphism
 \eqtx{
  {M }=\mat{0&0&1\\1&0&-5\\0&1&6}, 
 \textrm{ where } 
M V=V \Lambda  
 } 
 and choosing the initial lattice basis $L_0=V^{-1}$. After applying  $M^{n}$, the remapped lattice becomes
   \eqtx{ 
 \Tilde{L}_t= e^{tA} L_0 M_1^{n}= e^{tA}  \Lambda_1^{n}V^{-1}=e^{A_t}V^{-1},
 } 
 where $A_t=tA + n \log(\Lambda_1)$. This singe matrix is not enough to control the deformation. The  $LLL$ reduction algorithm~\cite{LLL}  is used to reduce the remapped lattice $\Tilde{L}_t$ by finding a matrix ${M_2}\in \SL$  using a high precision reduction, 
 \eqtx{ 
 \Hat{L_t}=LLL(\Tilde{L}_t)=  e^{A_t} V^{-1} {M_2}. 
 } 
In comparison to the GenKR's approach, such $M_2$ is automatically found on the earlier stage of the method while considering the communicative matrices. On this point, 
 we can improve the Hunt PBCs by finding the commutative matrix $M_2$ manually and apply the GenKR to produce remapped lattice which minimum distance is bounded before  the reduction step.
 The combination of this algorithm is presented in  Algorithms~\ref{GenKRAlg} will be presented later in the paper.  
 
\section{Rotating Box Algorithm}\label{sect3}
In this section, we will develop  PBCs for USF and BSF that are  time periodic up to a rotation. We write  the background flow as
\eqtx{
A=\epsilon D, \quad \textrm{where } D=\mat{
1  &0 & 0 \\
0 &1  & 0 \\
0 & 0 &-2 
}  .
}  
  Here, rather than choosing a pair of matrices $M_i\in \SL$ with real  spectrum, we will use a single matrix  $M\in \SL$ which has a pair of complex conjugate eigenvalues  and use it to remap the simulation box. 

Let us consider $M \in \SL$, a matrix that has a pair of complex conjugate eigenvalues  and write its real Jordan form 
\eqtx{ 
M V=V \Lambda   \textrm{ where } \Lambda  = \mat{\tilde{\eta} &-\tilde{\beta} &0\\\tilde{\beta} &\tilde{\eta} &0\\0&0&(\tilde{\eta} ^2+\tilde{\beta} ^2)^{-1}},  
}
 where $ \tilde{\eta} ,\tilde{\beta} \neq 0$ , and $\tilde{\eta} ^2+\tilde{\beta} ^2\neq 1$ in order to avoid a full rotation.
Taking the logarithm of $\Lambda$, we have
 \eqtx{ 
 \log( \Lambda) %
 =\mat{\eta  &- \beta   & 0 \\
  \beta  & \eta  & 0 \\
0 & 0 & -2\eta  }
   \textup{, where } \eta=\frac{1}{2}\log(\tilde{\eta}^2+\tilde{\beta}^2) ,  \beta  =  \arctan\Big(\frac{ \tilde{\beta}}{\tilde{\eta} }\Big) ,}
which can be decomposed as: 
\eqtx{  
 \log( \Lambda)=\beta  B +\eta D, \textrm{ where  }
B=\begin{bmatrix}
 0 & -1 & 0 \\
1&  0 & 0 \\
0 & 0 & 0 \end{bmatrix} .
}   
For all time $t$, by choosing the initial lattice $L_0=V^{-1}$, we can keep the lattice $ L_t  = e^{A t} L_0$ bounded by 
remapping the simulation box with $M^n$ 
 \eqtx{ 
  \Tilde{L}_t &=  e^{A t} L_0 M^{n} =   e^{\epsilon D t} \Lambda^{n}  V^{-1}  =e^{ n \beta  B }e^{A_t } V^{-1},
  } 
  where $A_t=\Big(\epsilon t- \mround{\frac{t \epsilon}{\eta}}\eta\Big) D$
for  $n=- \mround{\frac{t \epsilon}{\eta}}$. We have already seen in the planar elongational flow  case (Section~\ref{PEF}) that the remapped lattice $e^{A_t}V^{-1} $ is bounded and time-periodic of period  $ t_*=\frac{\eta}{\epsilon}$. 
In this case, the remapped lattice is time-periodic up to the rotation matrix $R=e^{n\beta  B}$. For a forward simulation in time, the algorithm reads 
\begin{algorithm}[H]\small
\caption{ R-KR
\label{CC-KR}}
\begin{algorithmic}
\State $ V,\Lambda =RealJordan(M)$ \Comment{Compute eigenbasis $V$ and the Jordan $\Lambda$ of $M$} 
\State $\eta D=\diag(\log( \Lambda ) )$\Comment{Compute the diagonal part of the logarithm of $\Lambda$}
\State $B=\log( \Lambda ) -\eta D$\Comment{Compute the rotation part of logarithm of $\Lambda$} 
\State $L_0= V^{-1}$\Comment{Compute the initial lattice }
\Statex
\State $\theta=0$ \Comment{Initialize $\theta$}
\For{$i = 1 \dots$ Nsteps}

\State  $n \leftarrow -\mround{ \frac{\theta }{\eta} + \tau}$ \Comment{Compute the power $n$ of $M$ necessary for the remap } 
 \State  $\theta\leftarrow \theta + \tau \eta +\eta n$   \Comment{Compute the remaining stretch value} 
\State  $\Tilde{L} \leftarrow   e^{n \beta  B} e^{ { \theta} D} V^{-1}$            \Comment{Compute the lattice value at the $t$ iteration } 
\EndFor
\end{algorithmic}  
\end{algorithm} 
Since the rotation matrix is bounded, we observe that the remapped lattice is also bounded during all the simulation. In the next section, we show that the rotating algorithm is not time periodic using the fact that the rotation matrix is never equal to the identity matrix for any automorphism chosen.
\section{Non time periodicity of the lattice in the rotating box algorithm  }\label{sect4} 

As mentioned above, for the class of automorphism matrices with real eigenvalues, it has been shown in \cite{KR} that it is impossible to construct KR PBCs with a time periodic lattice for USF or BSF. 
In this section, we will extend this demonstration to the class of automorphism matrices which have complex conjugate eigenvalues. Namely, we   show in the following corollary that although  the rotation  algorithm applied to USF or BSF is time-periodic up to a rotation matrix, there is no  choice of $M\in \SL$ where the period of the remapping   aligns with that of the rotation. In other word, we show that  rotation matrix $ e^{n \beta  B}$ is not equal to the identity matrix for $n\neq 0$, or $\beta$ is not a equal of $\pi$ times a rational number
\eqtx{ 
 \beta  =\tan^{-1}\Big(\frac{\tilde{\beta}}{\tilde{\eta}}\Big)\neq2 \pi\frac{ m}{n}, \quad n,m\in \mathbb{Z}, \quad n\neq 0,
}  
for any $M$ consider in Section~\ref{sect3}, i.e with complex eigenvalues one of the eigenvalue of  $M$ is not equal to  $1$.

We start by reminding that $\tilde{\eta}\pm \ii \tilde{\beta}, (\tilde{\eta}^{2}+\tilde{\beta}^{2})^{-1} $ are the roots of the  characteristic polynomial  $P(\lambda)=\lambda^3-h \lambda^2+k\lambda-1,h,k\in \mathbb{Z}$ of $M\in \SL$, and write these roots in the polar coordinate as $r^{-2}=(\tilde{\eta}^{2}+\tilde{\beta}^{2})^{-1} $, $re^{\pm\beta  }=\tilde{\eta}\pm \ii \tilde{\beta}$.
Let us first show the following lemma:
\begin{lemma}\label{pf}
   A matrix $M\in \SL$ with complex eigenvalues as define above  has $\beta  = 2 \pi\frac{ m}{n}$ if it has at least an  eigenvalue equal to $1$. 
\end{lemma} 
The proof of Lemma~\ref{pf} requires the use of the following results.
Let us consider  $\varphi$, the Euler totient function where $\varphi(n)$ is the number of positive integers that are relatively prime to $n$. A scalar $\alpha$ is said to be  algebraic over a field $K$ if there exists elements $a_0,\dots, a_i, (i\geq 1)$ of $K$, not equal to zero, such that 
\eqtx{
\alpha_0+ \alpha a_1+ \dots+\alpha^i  a_i=0,
}
and   $\deg\{\alpha\}$ is the degree of the irreducible characteristic polynomial. We refer the reader to  \cite[Chapter 4]{lang} or any introduction to Algebra book for the background about the definitions used in this section.
Then we have:
\begin{theorem}{\cite[Theorem 3.11]{Niven2}}\label{nivo} 
For $n>4$ and $\gcd(m,n)=1$, 
\eqtx{\deg\Big\{ \tan \frac{2 m \pi}{n} \Big\}=
\begin{cases}
\varphi(n) &\textrm{ for }\gcd(n,8)<4, \\
\frac{\varphi(n)}{2}   &\textrm{ for }  \gcd(n,8)=4, \\
\frac{\varphi(n)}{4}   &\textrm{ for } \gcd(n,8)>4.
\end{cases}
} 
\end{theorem} 
\begin{theorem}{\cite[Theorem 16.8.5]{artin}}\label{art}
For $K$ the splitting field of an irreducible cubic polynomial $P$ over a field $\mathbb{Q}$ and $D_P$ the discriminant of $P$,
\begin{itemize}
    \item If $D_P$ is a square in $\mathbb{Q}$, the degree of the extension field $K$ over $\mathbb{Q}$ is three 
        \item If $D_P$ is not a square in $\mathbb{Q}$, the degree of the extension field $K$ over $\mathbb{Q}$ is six. 
\end{itemize}  
\end{theorem} 
 We determine the degree of the  algebraic integer $\tan \beta  $  in the following lemma:
\begin{lemma}\label{degree}
   $\tan \beta  =\frac{\tilde{\beta}}{\tilde{\eta}}$ is an algebraic integer of degree at most six.
\end{lemma}
\begin{proof}
Since $\tilde{\eta}$ and $\tilde{\beta}$ are elements of the splitting field  $ K=\mathbb{Q}(r,e^{\beta  })$ of the irreducible polynomial $ P$, we have that  $\frac{\tilde{\beta}}{\tilde{\eta}}$ is also an element of $K$. By Theorem~\ref{art}, $K$ has a degree as most six in $\mathbb{Q}$ and so does  $\frac{\tilde{\beta}}{\tilde{\eta}}$. 
\end{proof} 
Let us prove  Lemma~\ref{pf} by finding all coefficients  $k,h\in \mathbb{Z}^+$ of the characteristic polynomial of  $M$ for which $\beta  =\frac{2 m \pi}{n}, m, n \in \mathbb{Z}.$ 
\begin{proof}  
Let us assume that
$\beta  =\frac{2 m \pi}{n}, m, n \in \mathbb{Z},$ and  find the possible $n,m$ by using  Theorem~\ref{nivo} and
the Theorem~\ref{art} which guarantee that $\tan \frac{2 m \pi}{n}$ is an algebraic integer of degree at most six.  
 Thus using \cite{Mendelsohn}, we find all $n$ that satisfy the following 
\eqtx{ 
\varphi(n) \leq 6  &\textrm{ for }\gcd(n,8)<4, \\
\frac{\varphi(n)}{2} \leq 6 &\textrm{ for }  \gcd(n,8)=4, \\
\frac{\varphi(n)}{4} \leq 6  &\textrm{ for } \gcd(n,8)>4,
}  
and report all $n$ and $\deg\Big\{\tan \frac{2 m \pi}{n}\Big\}\leq 6$ in   Table~\ref{tab}.
\begin{table}[h]
    \centering 
\begin{tabular}{r|l} 
$\deg\Big\{\tan \frac{2 m \pi}{n}\Big\}$ &$n$ \\ \hline
1 &1, 2   \\
2 &3, 6, 12, 16, 24 \\ 
4 &5, 10, 20, 32, 40, 48  \\
6 &7, 9, 14, 18, 28, 36, 56, 72 \\
\end{tabular} 
    \caption{Table of $n$ and $\deg\Big\{\tan \frac{2 m \pi}{n}\Big\}\leq 6$}
    \label{tab}
\end{table}
Then after few computing we find that
      \eqtx{ 
  h&=\frac{1+2r^3\cos{2 \pi \frac{m}{n}}}{r^2}, 
  k=\frac{r^3+2\cos{2\pi\frac{m}{n}}}{r},
  }  
  and plugging  in $n$ from the Table~\ref{tab} and  $m$ such that $\gcd(m,n)=1$,
  we remark that $k,h\in\mathbb{Z}^+$ if $n=1,2$. In result, $P$ has at least one eigenvalue  equal to   $1$, since $P(\lambda)=\lambda^3- \lambda^2+\lambda-1$, or $P(\lambda)=\lambda^3-3 \lambda^2+3\lambda-1$ for the latter values of $n$.
\end{proof} 
In sum, we derive the main  corollary of this section:
\begin{corollary}
  The rotating box algorithm cannot give a time-periodic simulation box for any choice of integer commutative complex conjugate matrix.
\end{corollary} 
\begin{proof}
Using Lemma~\ref{pf},  we know that only matrices with an eigenvalue equal to one 1 have a rotational part that is a root of unity.  
However, those matrices are themselves pure rotations and have no use for the PBCs since they cannot control the stretching caused by the underlying background flow.   
\end{proof}

\section{Comparison of the three dimensional algorithms } 
\label{sect5}
In this section, we  compute  the minimum distance of the particle images for our algorithm and compare it with the   GenKR using, Hunt's and Dobson's automorphism matrices. 

To compute the minimum distance between a particle and its images when the  rotating box PBCs is applied, we propose the matrix
\eqtx{
M&=\mat{0&-2&1\\1&1&0\\0&1&0}
,}  
which has  a pair of complex conjugate eigenvalues with positive real part.  
Then, the initial lattice is given by
\eqtx{
L_0= \mat{0.7726&0&-0.2083\\-0.26086&0.43442&0.48424\\-0.35555&-0.14106&0.84978},
} 
which implies that, given the standard lattice with the coordinate $(x,y,z)$, the $xy$ plane is rotated counterclockwise by approximately $113$ degrees, and $xz$ by $ 111 $ degrees. 

Moreover for the GenKR   algorithm, we keep the automorphism matrices and the initial lattice given in the original paper. The commutative matrices $M_i\in \SL$ and associated with the initial orthonormal lattice basis $L_0$ which determinant is equal to one, are given by   
\eqtx{
M_1=\mat{
1 & 1 & 1 \\
1  & 2 & 2 \\
1 & 2 & 3
}, 
M_2=\mat{
2 & -2 & 1 \\
-2 & 3 & -1\\
1 & -1 & 1
}\textrm{and, }L_0=\mat{0.59101&-0.73698&0.32799\\0.73698&0.32799&-0.59101\\0.32799&0.59101&0.73698}.
}   
For Hunt's formulation, we find a second automorphism matrix which   has positive eigenvalues and is commutative with the matrix given in the original paper.
The commutative matrices $M_1,M_2$ and the normalized initial lattice $L_0$   respectively read
 \eqtx{
 {M_1}=\mat{0&0&1\\1&0&-5\\0&1&6}, 
{M_2}=\mat{3&1&1\\-5&-2&-4\\1&1&4} 
\textrm{and, }L_0=\mat{0.52276&2.6394&13.3259\\0.52276&0.33619&0.2162\\0.52276&0.161&0.049584}.
}
\begin{algorithm}[H]\small
\caption{ GenKR-Hunt
\label{GenKRAlg}}
\begin{algorithmic} 
\State $ \Lambda_i=L_0 M_i L_0^{-1} $ \Comment{Diagonalization of $M_i$} 
\State $\hat{\omega}_i=\diag(\log(\Lambda_i))  $\Comment{Compute the logarithm of the diagonal of $\Lambda_i$ }
\State $A=\delta_1 \Hat{\omega}_1+ \delta_2 \Hat{\omega}_2    $\Comment{Compute $\delta_i$ } 
\Statex
\State $\theta_i=0$ \Comment{Initialize $\theta_i$}
\For{$i = 1 \dots$ Nsteps}   
  \State  $\theta_i\leftarrow \theta_i + \delta_k \tau t_* $   \Comment{Update the time} 
    \State  $\theta_i\leftarrow \theta_i -\mround{\theta_i} $   \Comment{Find the decimal part of $\theta_i$} 
 \State  $A_i\leftarrow \theta_1\Hat{\omega}_1 + \theta_2\Hat{\omega}_2 $   \Comment{Compute the remain stretch vector} 
\State  $\Tilde{L_i} \leftarrow    e^{ \diag(A_i)}L_0$            \Comment{Compute the lattice value at the $i$ iteration }
\State $\Hat{L_i} \leftarrow LLL(\Tilde{L}_i)$  \Comment{Reduce the remapped lattice with the $LLL$ reduction algorithm } 
\EndFor
\end{algorithmic}  
\end{algorithm}

    \begin{figure}
    \centering 
    {{\includegraphics[width=14cm]{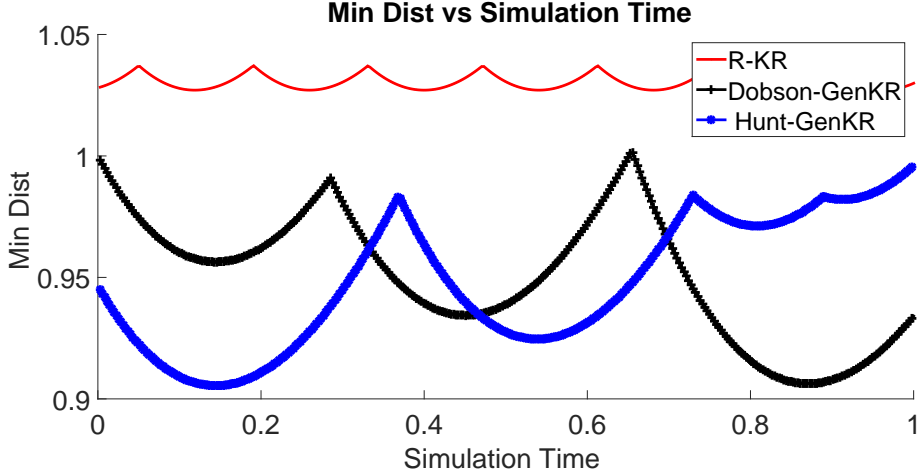} }} 
    \caption{Minimum distance vs simulation time for our new algorithm R-KR (red triangle), Dobson-GenKR (black plus) and, Hunt-GenKR  (blue stars). The minimum distance in our case is periodic and better than in the other cases. 
    }\label{compare11}
\end{figure}

Then we graph the  minimum distance for the three dimensional algorithms  in figure~\ref{compare11}, when  the stretch is $\epsilon=1$.  
We can observe in the graph that the  minimum distance curve in rotating box PBCs case presents a clear pattern of periodicity. In addition, the minimum distance for all the simulation for the rotating box algorithm is approximately 1.0271 compare to 0.9054 in GenKR's case.   
\section{Conclusion}\label{sect6} 
Kraynik-Reinelt proved that it is impossible to find  time periodic PBCs for general three dimensional flow, using $\SL$ matrices with real eigenvalues. In this paper, we show that by using an $\SL$ matrix with complex   eigenvalues, we can create an algorithm that is time-periodic   up to a rotation matrix for USF and BSF. Although we show that the rotations never align, the regularity of the remapping make this algorithm more straightforward than the existing ones. These PBCs also  offer a better minimum distance between a particle and its images.

\end{document}